\documentclass[12pt]{article}

\usepackage{eucal}
\usepackage{amssymb,latexsym,amsfonts,amsmath,amsthm,verbatim,hyperref,enumerate,dsfont}
\usepackage{color,cite,graphicx}
\usepackage{latexsym,amssymb,epsf} 

\usepackage{pstricks, pst-node}
\psset{arrows=->, labelsep=3pt, mnode=circle}
\input pst-plot

\DeclareMathAlphabet{\mathpzc}{OT1}{pzc}{m}{it}

\def\PP{\mathds{P}}
\def\EE{\mathds{E}}
\def\RR{\mathds{R}}
\def\NN{\mathds{N}}
\def\ZZ{\mathds{Z}}

\def\1{\mathds{1}}
\def\eps{\varepsilon}

\def\ToD{\overset{D}{\longrightarrow}}

\def\G{\mathcal{G}}
\def\E{\mathcal{E}}
\def\A{\mathcal{A}}
\def\Eq{\mathfrak{E}}
\def\Emb{\operatorname{Emb}}
\def\Proj{\operatorname{Proj}}
\def\Vol{\operatorname{Vol}}

\def\tr{\operatorname{tr}}

\def\supp{\operatorname{supp}}

\def\To{\longrightarrow}
\def\vv{\overrightarrow}
\def\<{\langle}
\def\>{\rangle}

\theoremstyle{plain}
\newtheorem{thm}{Theorem}[section]
\newtheorem*{thm*}{Theorem}
\newtheorem{lemma}[thm]{Lemma}
\newtheorem*{lemma*}{Lemma*}
\newtheorem{prop}[thm]{Proposition}
\newtheorem*{prop*}{Proposition}
\newtheorem{cor}[thm]{Corollary}
\newtheorem*{cor*}{Corollary}

\newtheorem*{cl*}{Claim}
\newtheorem*{obs*}{Observation}
\newtheorem{rmk}[thm]{Remark}
\theoremstyle{remark}
\newtheorem*{not*}{Notation}
\theoremstyle{definition}
\newtheorem{ex}[thm]{Example}
\newtheorem{dfn}[thm]{Definition}
\newtheorem*{qn*}{Question}

\numberwithin{equation}{section}

\begin{document}

\title{The spectral edge of some random band matrices}
\author{Sasha Sodin\footnote{Tel Aviv University, sodinale@post.tau.ac.il.
Supported in part by the Adams Fellowship Program of the Israel
Academy of Sciences and Humanities and by the ISF.}}
\maketitle

\begin{abstract}
We study the asymptotic distribution of the eigenvalues of random
Hermitian periodic band matrices, focusing on the spectral edges.
The eigenvalues close to the edges converge in distribution to the Airy
point process if (and only if) the band is sufficiently wide ($W \gg N^{5/6}$).
Otherwise, a different limiting distribution appears.
\end{abstract}


\section{Introduction}

In this paper, we study the edge of the spectrum of random Hermitian periodic
band matrices. The $N \times N$ Hermitian random matrix $H = H_N$ with rows and
columns labelled by elements of $\ZZ / N\ZZ$ has independent entries above
the main diagonal, and
\begin{equation}\label{eq:band}
H_{uv} = 0 \quad \text{if} \quad |u-v|_N = \min( |u-v|, N - |u-v|) > W_N
    \quad \text{or} \quad u=v~.
\end{equation}
To simplify the exposition, we assume that either
\begin{equation}\label{eq:pm}
\PP \{H_{uv} = 1 \} = \PP \{H_{uv} = -1 \} = 1/2~, \quad 0 < |u-v|_N \leq W_N
\end{equation}
(``random signs''), or
\begin{equation}\label{eq:ph}
H_{uv} = \exp(i U_{uv})~, \quad U_{uv} \sim U(0, 2\pi)~, \quad 0 < |u-v|_N \leq W_N
\end{equation}
(``random phases''), and defer the discussion of possible generalisations to the last
section. For the same reason, we assume that $W = W_N \to \infty$ as $N \to \infty$.
\\*
\\*
The matrix $H_N$ is closely related to the graph $\G = (\ZZ/N\ZZ, \E)$, where
\begin{equation}\label{eq:def.g}
(u, v) \in \E \iff 0 < |u-v| \leq W_N~,
\end{equation}
and can be viewed as a Hamiltonian of quantum evolution in a disordered environment
on $\G$ ((\ref{eq:ph}) corresponds to broken time-reversal symmetry.) We refer the
reader to the work of Fyodorov and Mirlin \cite{FM} for a thorough discussion of
physical motivation.
\\*
\\*
Random matrices similar to $H_N$ have been studied in both mathematical and physical
literature. We survey some of the results that pertain to the current work, and refer
to \cite{FM} and to the recent review of Spencer \cite{Sp} for detailed bibliography.

\newcounter{remnr1}
\def\rem1{
    \addtocounter{remnr1}{1}
    \vspace{2mm}\noindent{\bf (\Alph{remnr1})} }

\rem1 Bogachev, Molchanov, and Pastur \cite{BMP} have proved that the empirical spectral
measure (or integrated density of states)
\begin{equation}\label{eq:def.mu}
\mu_N(\alpha) = \# \left\{ \text{eigenvalues of $\frac{H_N}{2\sqrt{2W_N}}$
    in $(-\infty, \alpha]$} \right\}
\end{equation}
converges (weakly, in distribution, as $N \to \infty$ and $W_N \to \infty$), to the
(deterministic) Wigner measure
\begin{equation}\label{eq:w}
\mu_{\operatorname{Wigner}}(\alpha)
    = \int_{-\infty}^\alpha \frac{2}{\pi} (1-\alpha_1^2)_+^{1/2} d\alpha_1~.
\end{equation}
That is, the {\em global} behaviour of the spectrum of $H_N$ is similar to that of Wigner
matrices (which correspond to the special case $W_N = N/2$.)

\rem1 One of the interesting questions concerning {\em local} eigenvalue statistics is the
crossover between the Random Matrix regime and the Poisson regime. The Thouless criterion
\cite{Th}, applied to random band matrices by Fyodorov and Mirlin \cite{FM}, predicts the
following.\\*
If the mixing exponent $\rho_\text{mix}$ of the (classical) random walk on $\G$ is much
larger than the eigenvalue spacing near $\alpha_0$:
\[ \rho_\text{mix} \gg \frac{\text{mean spacing at $\alpha_0$}}
                            {\text{mean density at $\alpha_0$}}~, \]
then the eigenvalues of $H_N$ near $\alpha_0$ obey Random Matrix statistics, and the
corresponding eigenvectors are extended. If
\[ \rho_\text{mix} \ll \frac{\text{mean spacing at $\alpha_0$}}
                            {\text{mean density at $\alpha_0$}}~, \]
the eigenvalues near $\alpha_0$ obey Poisson statistics, and the corresponding
eigenvectors are localised. \\*
For our graph $\G$, $\rho_\text{mix}$ is of order $W_N^2/N^2$, and in the bulk
of the spectrum, $-1 < \alpha_0 < 1$, (\ref{eq:w}) suggests that
\[ \frac{\text{mean spacing at $\alpha_0$}}
        {\text{mean density at $\alpha_0$}} \asymp 1/N \]
Thus the crossover should occur at $W_N \asymp \sqrt{N}$. \\*
On the physical level of rigour, these predictions have been justified by Fyodorov
and Mirlin \cite{FM}, who have derived a detailed description of both asymptotic
regimes (and also of the crossover.) So far, these results resist mathematical
justification (see however the work of Khorunzhiy and Kirsch \cite{KK} for some
relatively recent developments.) \\*
We remark that Spencer and Wang \cite[Ch.\ III]{W} have managed to make one direction of
(a slightly different form of) the Thouless criterion rigorous under certain assumptions,
in a fairly general setting. However, for now these assumptions have not been verified
for the problem under consideration. We refer the reader to the review of Spencer \cite{Sp}
for a discussion of the mathematical approach to the Thouless criterion and its application
to band matrices.

\rem1 We focus on the edge of the spectrum: $\alpha_0 = \pm 1$. There,
\[ \frac{\text{mean spacing at $\alpha_0$}}
        {\text{mean density at $\alpha_0$}} \asymp \frac{N^{-2/3}}{N^{-1/3}} = N^{-1/3}~,\]
therefore the crossover should occur at $W_N \asymp N^{5/6}$. On the physical level of
rigour, this has been confirmed by Silvestrov \cite{Si}. In fact, one of the aims of the
current paper is to put some of the methods and results of \cite{Si} on firm mathematical
ground.

\vspace{3mm}\noindent
Now we state the main results.
\newcounter{remnr2}
\def\rem2{
    \addtocounter{remnr2}{1}
    \vspace{2mm}\noindent{\pscirclebox{\bf \arabic{remnr2}}} }

\rem2 For $W_N \gg N^{5/6}$ (i.e.\ $W_N/N^{5/6} \to +\infty$), set
\[\begin{split}
    \sigma_R(\lambda) &= \# \left\{ \text{eigenvalues of $\frac{H_N}{2\sqrt{2W_N}}$
        in $\big(1 - \frac{\lambda}{2N^{2/3}}, \, +\infty \big]$} \right\}~, \\
    \sigma_L(\lambda) &= \# \left\{ \text{eigenvalues of $\frac{H_N}{2\sqrt{2W_N}}$
        in $\big(-\infty, \, -1+ \frac{\lambda}{2N^{2/3}} \big]$} \right\}~.
\end{split}\]
\begin{thm}\label{th:RM}
If $W_N / N^{5/6} \to \infty$ as $N \to  \infty$, the measures $\sigma_R(\lambda)$ and
$\sigma_L(\lambda)$ converge in distribution\footnote{on test functions $f \in C(\RR)$
such that $\supp f \cap \RR_+$ is compact} to $\mathfrak{Ai}_\beta(-\lambda)$,
where $\mathfrak{Ai}_\beta$ is the (random) distribution function corresponding
to the Airy point process, $\beta=1$ for (\ref{eq:pm}), and $\beta = 2$ for
(\ref{eq:ph}).

Hence, according the Tracy--Widom theorem \cite{TW1,TW2}, the scaled extreme
eigenvalues $2N^{2/3} (\alpha_{\max}-1)$ and $-2N^{2/3} (\alpha_{\min}+1)$
of $H_N/(2\sqrt{2W_N})$ converge in distribution to the Tracy--Widom law $TW_\beta$.
\end{thm}

\rem2 For $1 \ll W_N \ll N^{5/6}$, set
\[\begin{split}
    \sigma_R(\lambda) &= \frac{W_N^{6/5}}{N} \,
        \# \left\{ \text{eigenvalues of $\frac{H_N}{2\sqrt{2W_N}}$
            in $\big(1 - \frac{\lambda}{2W_N^{4/5}}, \, +\infty \big]$} \right\}~, \\
    \sigma_L(\lambda) &= \frac{W_N^{6/5}}{N} \,
        \# \left\{ \text{eigenvalues of $\frac{H_N}{2\sqrt{2W_N}}$
            in $\big(-\infty, \, -1+ \frac{\lambda}{2W_N^{4/5}} \big]$} \right\}~.
\end{split}\]

\begin{thm}\label{th:P}
If $W_N \to \infty$ and $W_N/N^{5/6} \to 0$, the measures $\sigma_R$ and $\sigma_L$
converge in distribution to the (same) deterministic measure $\sigma_\beta$. We have:
\[\begin{split}
&\sigma_\beta(\lambda) = \frac{2}{3\pi} \lambda^{3/2} + O(\lambda)~,
    \quad \lambda \to +\infty~, \\
&\sigma_\beta(\lambda) \leq C \exp(-C |\lambda|^{5/4})~, \quad \lambda \to -\infty~.
\end{split}\]
\end{thm}

\vspace{1mm}\noindent In some sense, the eigenvalues close to the
edges behave as independent random samples from $\sigma_\beta$. The
method of the current paper yields the following manifestation of
this belief:
\[ \PP \left\{ 2W_N^{4/5}(1 - \alpha_{\max}) \geq \lambda \right\}
    - \exp \left\{ - \frac{N}{W_N^{6/5}} \sigma_\beta(\lambda) \right\} = o(1) \]
uniformly in $\lambda \in \RR$. Unfortunately, our description of the left tail of
$\sigma_\beta$ is not sufficiently precise to deduce convergence to a max-stable law
(cf.\ Gnedenko \cite{G} for the description of domains of attraction of max-stable laws.)

\vspace{1mm}\noindent
Khorunzhiy \cite{K} has proved (for a slightly different class of band matrices) that,
if $W_N \gg \log^{3/2} N$,
\begin{equation}\label{eq:norm}
\left\| H_N \big/ (2 \sqrt{2W_N}) \right\| \ToD 1
\end{equation}
as $N \to \infty$ (actually, he has established a stronger form of convergence.)
He has conjectured (private communication) that the same conclusion
holds under the weaker assumption $W_N \gg \log N$. We confirm this conjecture:
\begin{thm}\label{th:norm}
If $W_N / \log N \to +\infty$, then $\left\| H_N \big/ (2 \sqrt{2W_N}) \right\| \ToD 1$.
\end{thm}
\noindent As one can see from the argument of Bogachev, Molchanov, and Pastur \cite{BMP},
this result is sharp, meaning that the conclusion fails if $W_N / \log N \to 0$.

\vspace{4mm}\noindent
The proofs of the three results are based on a modification of the moment method. The
moment method has been applied to random matrices since the work of Wigner \cite{W2}.
Bogachev, Molchanov, and Pastur \cite{BMP} have applied it to study the spectrum of band
matrices (see above.)

The moment method appears particularly useful to study the eigenvalue statistics at the
spectral edge. In \cite{S1}, Soshnikov has applied it to derive the limiting distribution
of the extreme eigenvalues of Wigner-type random matrices. Extensions of his approach have
allowed to solve a number of related problems, e.g.\ \cite{S2,P,FP}.

In this paper, we apply a modification of the moment method, which goes back at least to
the work of Bai and Yin \cite{BY} (see \cite{me} for more detailed references.) In \cite{FS},
it has been used, in particular, to give another proof of Soshnikov's result. The
combinatorial technique of \cite{FS} is the main ingredient of the current work. In
Section~\ref{s:p}, we review this technique, and formulate the combinatorial statements
needed to prove the main results.

\vspace{2mm}\noindent
Another ingredient of the proof is an asymptotic description of the (classical) random walk
on $\G$. We prove the necessary facts in Section~\ref{s:rw}. In Section~\ref{s:emb}, we apply
these facts to count subgraphs of $\G$ of a certain form. In Section~\ref{s:p.comb}, we
specialise to the setting of Section~\ref{s:p}, and prove the combinatorial statements
formulated there.

\vspace{2mm}\noindent
Section~\ref{s:sl} collects some facts related to Levitan's uniqueness theorem \cite{L}
which we use in the sequel.

\vspace{2mm}\noindent
In Section~\ref{s:pr}, we conclude the proofs of the main results. Section~\ref{s:ph}
explains how to modify the proofs written for (\ref{eq:pm}) to deal with (\ref{eq:ph}).
In Section~\ref{s:cr}, we discuss generalisations and related open problems.

\vspace{4mm}\noindent
{\bf Notation:} In this paper, $C, c, \cdots$ stand for positive constants, the value
of which may change from line to line. If $X, Y$ are quantities depending on some large
parameter, $X \ll Y \iff Y \gg X \iff X = o(Y) \iff X/Y \to 0$.

\section{Preliminaries}\label{s:p}

From this section, we omit the subscript $N$, and write $H$ for $H_N$ and $W$ for $W_N$.
Also, we consider the matrices with entries (\ref{eq:pm}); in Section~\ref{s:ph} we shall
explain the modifications needed for the case (\ref{eq:ph}).

Let $U_n(\cos \theta) = \sin((n+1)\theta)/\sin \theta$ be the Chebyshev polynomials of the
second kind; set $U_{-2} \equiv U_{-1} \equiv 0$. Let
\[ H^{(n)} = (2W-1)^{n/2} \left\{ U_n \left( \frac{H}{2\sqrt{2W-1}} \right)
    - \frac{1}{2W - 1} \, U_{n-2} \left( \frac{H}{2\sqrt{2W-1}} \right) \right\}~.\]
The following lemma (see e.g.\ \cite{me} for a more general version) is at the basis of our
considerations.

\begin{lemma}
For any Hermitian $N \times N$ matrix $H$ satisfying
\[ H_{uv} = \begin{cases}
    \pm 1,  &0 < |u-v|_N \leq W \\
    0,      &\text{otherwise}
    \end{cases}~,\]
and any $u_0, u_n \in \ZZ/N\ZZ$,
\[ H^{(n)}_{u_0u_n} = {\sum}' H_{u_0u_1} H_{u_1u_2} \cdots H_{u_{n-1}u_n}~,\]
where the sum is over all $(n-1)$-tuples $(u_1, u_2, \cdots, u_{n-1})$ (which we regard
as paths $p_n = u_0 u_1 \cdots u_n$), such that
\begin{description}
\item[(a)] $0 < |u_j - u_{j+1}|_N \leq W$, $0 \leq j \leq n-1$ ($p_n$ is a path on $\G$),
\item[(b)] $u_{j+2} \neq u_j$, $0 \leq j \leq n-2$ ($p_n$ is {\em non-backtracking}.)
\end{description}
\end{lemma}

\noindent The following corollary is immediate from the lemma and (\ref{eq:pm}).

\begin{cor}\label{cor} For the random matrix $H$ as defined in the introduction,
\[ \EE \tr H^{(n(1))} \tr H^{(n(2))} \cdots \tr H^{(n(k))} \]
is equal to the number of $k$-tuples of paths (shortly: $k$-paths)
\[ p_{n(1),\cdots,n(k)} = u_0^1 u_1^1 \cdots u_{n(1)}^1 \,\,
                          u_0^2 u_1^2 \cdots u_{n(2)}^2 \,\, \cdots \,\,
                          u_0^k u_1^k \cdots u_{n(k)}^k \]
that satisfy (a),(b), and also
\begin{description}
\item[(c)] $u_{n(j)}^j = u_0^j$, $1 \leq j \leq k$ (the paths are closed);
\item[(d)] the number
\[ \# \left\{ (i, j) \, | \, u_i^j = u, \, u_{i+1}^j = v \right\}
    - \# \left\{ (i, j) \, | \, u_i^j = v, \, u_{i+1}^j = u \right\} \]
is even, for any $u,v \in \ZZ/N\ZZ$.
\end{description}
\end{cor}

\noindent As in \cite{FS}, we group the $k$-paths satisfying (a)-(d) into topological
equivalence classes, which are in one-to-one correspondence with {\em $k$-diagrams}:

\begin{dfn}\label{def:diag.k} Let $\beta \in \{1,2\}$.\hfill
A {\em $k$-diagram} is an (undirected) multigraph $\bar{G} = (\bar{V}, \bar{E})$,
together with a $k$-tuple of circuits
\begin{equation}\label{eq:barp} \bar{p} = \bar{u}_0^1 \bar{u}_1^1 \cdots \bar{u}_0^1 ,\,\,\,
             \bar{u}_0^2 \bar{u}_1^2 \cdots \bar{u}_0^2 ,\,\,\,
             \cdots ,\,\,\, \bar{u}_0^k \bar{u}_1^k \cdots \bar{u}_0^k
\end{equation}
on $\bar{G}$, such that
\begin{itemize}
\item $\bar{p}$ is  {\em non-backtracking} (meaning that in every circuit no edge is
followed by its reverse, unless the edge is a loop $\bar{u}\bar{u}$);
\item For every $(\bar{u}, \bar{v}) \in \bar{E}$,
\[ \# \left\{ (i,j) \, | \, \bar{u}_j^i = \bar{u}, \, \bar{u}_{j+1}^i = \bar{v} \right\}
+ \# \left\{ j \, | \, \bar{u}_j^i = \bar{v}, \, \bar{u}_{j+1}^i = \bar{u} \right\} = 2~;\]
\item the degree of $u_0^i$ in $\bar{G}$ is 1; the degrees of all the other vertices
are equal to 3.
\end{itemize}
\end{dfn}

\begin{rmk}\label{rem:clar} We emphasise that $\bar{G}$ is a multigraph in which the
coinciding edges are distinguished. Thus, strictly speaking, a circuit is not uniquely
determined by the vertices it passes. Still, we find it convenient to use the notation
as in (\ref{eq:barp}). Next, we do not distinguish two diagrams which are isomorphic
in the natural sense. Thus, by a diagram we actually mean an equivalence class (e.g.\
in the second part of the following lemma.)
\end{rmk}

\noindent The following lemma summarises some properties of $k$-diagrams from \cite[Part~II]{FS}.

\begin{lemma}\label{l:fs}\hfill
\begin{enumerate}
\item For every $k$-diagram, there exists an integer $s \geq k$ (``non-orientable genus'', cf.\
Figures 1,2), such that the diagram has $2s$ vertices and $3s-k$ edges.
\item The number $D_k(s)$ of $k$-diagrams corresponding to a given $s$ satisfies
\[ \frac{(s/C)^{s+k-1}}{(k-1)!} \leq D_k(s) \leq \frac{(Cs)^{s+k-1}}{(k-1)!}\]
\end{enumerate}
\end{lemma}

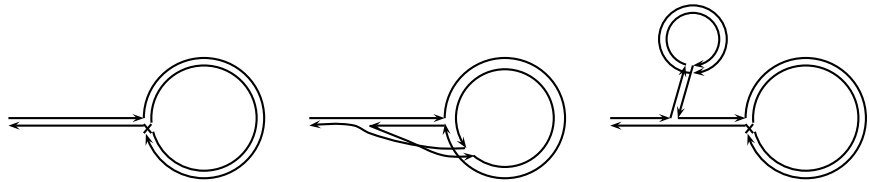
\begin{figure}[h]
\vspace{2cm}
\setlength{\unitlength}{1cm}
\begin{pspicture}(-1,0)
\psline(.2,1)(2, 1)
\psarcn(2.8,1){.8}{180}{195}
\psset{arrows=-}
\psline(2,.8)(2.1,.92)
\psarcn(2.8,1){.7}{185}{195}
\psline(2,.92)(2.1,.8)
\psset{arrows=->}
\psline(2,.9)(.2,.9) %
\psline(4.2,1)(6,1)
\psarcn(6.8,1){.8}{180}{185}
\psline(6,.9)(5,.9)
\pscurve(5,.9)(5.5,.7)(6,.5)(6.4,.5)
\psarc(6.8,1){.65}{230}{218}
\pscurve(6.27,.6)(5.9, .6)(5,.8)(4.8,.9)(4.2,.9) %
\psline(8.2,1)(9, 1)
\psline(9, 1)(9.2, 1.7)
\psarcn(9.3, 2.05){.35}{255}{270}
\psarcn(9.3, 2.05){.45}{265}{270}
\psline(9.3, 1.7)(9.1, 1)
\psline(9.1, 1)(10, 1)
\psarcn(10.8,1){.8}{180}{195}
\psset{arrows=-}
\psline(10,.8)(10.1,.92)
\psarcn(10.8,1){.7}{185}{195}
\psline(10,.92)(10.1,.8)
\psset{arrows=->}
\psline(10,.9)(8.2,.9)
\end{pspicture}
\caption[Fig. 4.5]{Some 1-diagrams: $s = 1$ (left),
    $s = 2$ (center, right)}\label{fig:diag1}
\vspace{2mm}
\end{figure}

\begin{figure}[h]
\vspace{3cm}
\setlength{\unitlength}{1cm}
\begin{pspicture}(-1,0)
\psline(.2,1)(1, 1)
\psarcn(1.6,1){.6}{180}{187}
\psline(1,.9)(.2,.9)
\psline(3,1)(2.05, 1)
\psarc(1.6,1){.5}{0}{352}
\psline(2.05, .9)(3,.9)
\psline(4,0)(4, 1)
\psarcn(4,1.5){.5}{270}{278}
\psset{arrows=-}
\psline(4.1,1)(4,1.1)
\psarcn(4,1.5){.4}{268}{282}
\psline(4.1,1.1)(4,1)
\psset{arrows=->}
\psline(4.1,1)(4.1,0)
\psline(4.7,2)(4.7, 1)
\psarc(4.7,.5){.5}{90}{82}
\psset{arrows=-}
\psline(4.8,1)(4.7,.9)
\psarc(4.7,.5){.4}{90}{82}
\psline(4.8,.9)(4.7,1)
\psset{arrows=->}
\psline(4.8,1)(4.8,1.5)
\psline(4.8,1.5)(5.8, 1.5)
\psarc(6.4,1.5){.6}{180}{173}
\psset{arrows=-}
\psline(5.8, 1.6)(5.9, 1.5)
\psarc(6.4,1.5){.5}{180}{170}
\psline(5.9,1.6)(4.8, 1.6)
\psset{arrows=->}
\psline(4.8,1.6)(4.8,2)
\psline(7.7,1)(8.5, 1)
\psline(8.5, 1)(8.7, 1.7)
\psarcn(8.8, 2.05){.35}{255}{270}
\psarcn(8.8, 2.05){.45}{265}{270}
\psline(8.8, 1.7)(8.6, 1)
\psline(8.6, 1)(9.5, 1)
\psarcn(10.1,1){.6}{180}{190}
\psline(9.5,.9)(7.7,.9)
\psline(11.1, 2)(10.45, 1.35)
\psarc(10.1,1){.5}{45}{37}
\psline(10.5, 1.25)(11.17, 1.93)
\end{pspicture}
\caption[Fig. 4.5]{Some 2-diagrams: $s = 2$ (left, center),
    $s = 3$ (right)}\label{fig:diag2}
\end{figure}
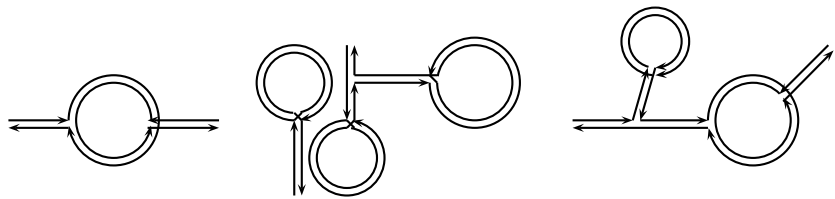

\noindent Thus our goal is to compute the number of $k$-paths corresponding to a given diagram.
It will be convenient to consider connected diagrams only. This can be done as follows: set
\[\begin{split}
T(n(1))         &= \EE \tr H^{(n(1))}~, \\
T(n(1), n(2))   &= \EE \tr H^{(n(1))} \tr H^{(n(2))} - T(n(1)) T(n(2))~, \\
&\cdots\cdots \\
T(n(1), \cdots, n(k)) &= \EE \prod_{j=1}^k H^{(n(j))}
    - \sum_\Pi \prod_{P \in \Pi} T(\{n(j)\}_{j \in P})~,
\end{split}\]
where the sum is over non-trivial partitions of $\{1,2,\cdots,k\}$ into disjoint sets.
It is not hard to see that indeed $T(n(1), \cdots, n(k))$ counts the number of $k$-paths
$p_{n(1),\cdots,n(k)}$ that satisfy (a)--(d) and correspond to connected diagrams.
\\*
\\*
Now we formulate the two main technical statements of this paper.

\begin{prop}\label{p:main.rm} Let $W \gg N^{5/6}$, $R \geq 0$. For any
\[ 1 \leq n(1) \leq n(2) \cdots \leq n(k) \leq R N^{1/3} \]
such that $n(1) + \cdots + n(k) = 2n$ is even,
\begin{multline*}
\frac{T(n(1), \cdots, n(k))}{(2W-1)^n} = \\
\left[ \prod_{j=1}^k n(j) \right] \, \phi_k(n(1)/N^{1/3}, \cdots, n(k)/N^{1/3})
    + N^{k/3} \eps_N(n(1), \cdots, n(k))~,
\end{multline*}
where
\begin{enumerate}
\item $\phi_k(z_1, \cdots, z_k) = \sum_{s \geq k} g_{k, s}(z_1, \cdots, z_k)$,
$g_{k,s}$ being a continuous homogeneous function of degree $3(s-k)$,
\[ g_{k,s}(z_1, \cdots, z_k) \leq \frac{(C\|z\|)^{3(s-k)}}{(cs)^{2s-3k+1}}~.\]
\item $\eps_N(n(1), \cdots, n(k)) = o(1)$, where the implicit constant depends only
on $R$ and on $W/N^{5/6}$, and
\[ \sum_{1 \leq n(1) \leq \cdots \leq n(k) \leq N^{1/3}}
    \frac{\eps_N(n(1), \cdots, n(k))}{n(1) \cdots n(k)} \leq C_k.\]
\end{enumerate}
If $n(1) + \cdots + n(k) = 2n+1$ is odd, $T(n(1), \cdots, n(k)) = 0$.
\end{prop}

\begin{prop}\label{p:main.rw}
Let $W \ll N^{5/6}$, $R \geq 0$. For any
\[ 1 \leq n(1) \leq \cdots \leq n(k) \leq R W^{2/5} \]
such that $n(1) + \cdots + n(k) = 2n$ is even,
\begin{multline*}
\frac{T(n(1), \cdots, n(k))}{(2W-1)^n} = \\
\frac{N}{W^{6/5}} \, \left\{ \left[ \prod_{j=1}^k n(j) \right]\,
    \psi_k\left(\frac{n(1)}{W^{2/5}}, \cdots, \frac{n(k)}{W^{2/5}}\right)  + W^{2k/5}
        \eps(n(1), \cdots, n(k))\right\}~,
\end{multline*}
where
\begin{enumerate}
\item $\psi_k(z_1, \cdots, z_k) = \sum_{s \geq k} h_{k, s}(z_1, \cdots, z_k)$,
$h_{k,s}$ being a homogeneous function of degree $(5s-5k-1)/2$, continuous outside
the origin,
\[ h_{k,s}(z_1, \cdots, z_k) \leq \frac{(C\|z\|)^\frac{5s-5k-1}{2}}{(cs)^\frac{3s-5k+1}{2}}~.\]
\item $\eps_N(n(1), \cdots, n(k)) = o(1)$, where the implicit constant depends only
on $R$ and on $N^{5/6}/W$, and
\[ \sum_{1 \leq n(1) \leq \cdots \leq n(k) \leq W^{2/5}}
    \frac{\eps_N(n(1), \cdots, n(k))}{n(1) \cdots n(k)} \leq C_k.\]
\end{enumerate}
If $n(1) + \cdots + n(k) = 2n+1$ is odd, $T(n(1), \cdots, n(k)) = 0$.
\end{prop}

\section{Random walk on the circle}\label{s:rw}

Let $\G = (\ZZ/ N \ZZ, \E)$ be defined by (\ref{eq:def.g}) (in this section we do
not assume that $W \gg 1$.) Denote by $\mathcal{W}_n(R)$ the number of paths of length
$n$ in $\G$ between two (fixed) vertices $u,v$ such that $|u-v|_N = R$. We prove the
following statements.

\begin{prop}\label{p:clt1} If $1 \ll n \ll N^2/W^2$, $R \ll n^{3/4} W$, and $R \leq 0.49 N$, then
\[ \frac{\mathcal{W}_n(R)}{(2W)^n} = (1+o(1))
    \left[ \frac{\pi n}{3}(W+1) (2W+1) \right]^{-1/2}
    \exp \left\{ - \frac{3R^2}{n (W+1)(2W+1)} \right\}~.\]
\end{prop}

\begin{prop}\label{p:clt2}
If $N^2/W^2 \ll n$,
\[ \frac{\mathcal{W}_n(R)}{(2W)^n} = (1+o(1)) N^{-1}~. \]
\end{prop}

\begin{prop}\label{p:clt3}
Without any assumptions on $n, N, R, W$,
\[ \frac{\mathcal{W}_n(R)}{(2W)^n} \leq C \left[
    (W\sqrt{n})^{-1} \exp\left\{- \frac{CR^2}{nW^2}\right\}  + N^{-1} \right]~. \]
\end{prop}

\noindent For fixed $W$, Proposition~\ref{p:clt1} follows from Richter's local limit theorem for
lattice variables (see Ibragimov and Linnik, \cite[Ch.~VII]{IL}.) However, we need
the asymptotics to be uniform in $W$, which we have not found in the literature.
Proposition~\ref{p:clt2} (perhaps with an extra logarithmic factor in the assumptions)
follows easily from the spectral estimates on the mixing time.
\\*
\\*
Let $\A$ be the adjacency matrix of $\G$. As $\G$ is invariant under cyclic shifts,
the discrete Fourier transform diagonalises $\A$. We state this as a lemma:

\begin{lemma}\label{l:gspec}
$(2W)^{-1}\A = \sum_{k=0}^{N-1} a_k
    \mathfrak{f}_k \otimes \mathfrak{f}_k$,
where
\[ a_k = \frac{\sin \left( W \frac{\pi k}{N} \right)}
                                            {W \sin \frac{\pi k}{N}}
                                        \cos\left( (W+1) \frac{\pi k}{N} \right)~,  \]
and $\mathfrak{f}_k(\ell) = \exp \left\{ \frac{2\pi i k \ell}{N} \right\}$.
\end{lemma}

\begin{proof}
The vectors $\mathfrak{f}_k$ are of unit length, hence we only need to check that
$\A \mathfrak{f}_k = 2 W a_k \mathfrak{f}_k$. Let $\omega = \exp \frac{2\pi i}{N}$. Then
\[ (\A \mathfrak{f}_k)(m)
    = \sum_{0 < |\ell - m|_N \leq W} \omega^{k\ell}
    = \mathfrak{f}_k(m) \sum_{0 < |\ell|_N \leq W} \omega^{k\ell}~.\]
Now,
\[\begin{split}
\sum_{0 < |\ell|_N \leq W} \omega^{k\ell}
    &= -1 + \omega^{-kW} \frac{1 - \omega^{(2W+1)k}}{1 - \omega^k} \\
    &= \frac{\omega^{(W+1/2)k} - \omega^{k/2} + \omega^{-k/2} - \omega^{-(W+1/2)k}}
            {\omega^{k/2} - \omega^{-k/2}} \\
    &= \frac{\omega^{Wk/2} - \omega^{-Wk/2}}{\omega^{k/2} - \omega^{-k/2}}
        \left(\omega^{(W+1)k/2} + \omega^{-(W+1)k/2}\right) = 2W a_k~.
\end{split}\]
\end{proof}

\noindent Denote by $\delta_0, \cdots, \delta_{N-1}$ the standard basis in $\RR^{\ZZ/N\ZZ}$.
Then
\[ \delta_j = N^{-1/2} \sum_{k=0}^{N-1} \exp \frac{-2\pi ijk}{N} \,\, \mathfrak{f}_k~,\]
therefore
\begin{equation}\label{eq:wn}
\frac{\mathcal{W}_n(R)}{(2W)^n}
    = \<(\A/2W)^n \delta_R, \, \delta_0 \>
    = N^{-1} \sum a_k^n \exp \frac{2\pi i Rk}{N}~.
\end{equation}
Thus we have to estimate the above sum. Informally, the argument is as follows:
$\max(|a_1|, \cdots, |a_{N-1}|) = 1 - \Theta(N^2/W^2)$, whereas $a_0 = 1$. If $n \gg N^2 / W^2$,
the sum is dominated by the first addend, which is equal to 1. If $n \ll N^2 / W^2$, the sum
can be replaced with an integral, which is then evaluated using the saddle point method. And
now to the formal proof.

\vspace{1mm}\noindent Denote
\[ f(z) = \frac{\sin \left[ W \pi z \right]}{W \sin \left[ \pi z \right]}
    \, \cos \left[ (W+1) \pi z \right]~.\]
The following lemma summarises some properties of $f$:
\begin{lemma}\label{l:f.decay}
The function $f$ is an entire function with period $1$;
\[ \frac{|f(x+iy)|}{|f(iy)|} \leq \exp \left\{ - c W^2 x^2 \right\}~,
    \quad -1/2 \leq x \leq 1/2~. \]
\end{lemma}

\begin{proof}[Proof of Proposition~\ref{p:clt1}]
Choose $\frac{\sqrt{n} W}{N} \ll \eta \ll 1$, and consider two cases.

\vspace{2mm}\noindent
(a): $R \leq \eta \sqrt{n} W$. Then, according to (\ref{eq:wn}) and Lemma~\ref{l:f.decay},
\[\begin{split}
\frac{\mathcal{W}_n(R)}{(2W)^n}
    &= N^{-1} \sum_{-N/2 < k \leq N/2} \big[ f(k/N) \big]^n \exp \frac{2\pi i Rk}{N} \\
    &= N^{-1} \sum_{-N/2 < k \leq N/2} \big[ f(k/N) \big]^n
        \left[ 1 + O \left( \frac{\eta |k| \sqrt{n} W}{N} \right) \right]~.
\end{split}\]
The contribution of the second addend is negligible:
\begin{multline*}
\left| N^{-1} \sum_{-N/2 < k \leq N/2} \big[ f(k/N) \big]^n
        \, \frac{\eta |k| \sqrt{n} W}{N} \right| \\
    \leq \left| N^{-1} \sum_{-N/2 < k \leq N/2} \exp \left\{ - \frac{cW^2 n k^2}{N^2} \right\}
        \frac{\eta |k| \sqrt{n} W}{N} \right|
    \leq \frac{C \eta}{\sqrt{n} W} \ll \frac{1}{\sqrt{n} W}~.
\end{multline*}
Then,
\begin{multline*}
N^{-1} \sum \big[ f(k/N) \big]^n
    = N^{-1} \sum \exp \left\{ - \frac{\pi n}{3} (W+1)(2W+1) \frac{k^2}{N^2} \right\} \\
    + N^{-1} \sum \left( \big[ f(k/N) \big]^n
        - \exp \left\{ - \frac{\pi n}{3} (W+1)(2W+1) \frac{k^2}{N^2} \right\}\right)~.
\end{multline*}
The second sum is negligible, whereas the first one is
\begin{multline*}
(1 + o(1)) N^{-1} \sum_{k=-\infty}^\infty
    \exp \left\{ - \frac{\pi n}{3} (W+1)(2W+1) \frac{k^2}{N^2} \right\} \\
    = (1 + o(1)) \, \left[ \frac{\pi n}{3} (W+1) (2W+1) \right]^{-1/2}~,
\end{multline*}
according to the Poisson summation formula.

\vspace{2mm}\noindent
(b): $R > \eta \sqrt{n} W$. According to (\ref{eq:wn}), Lemma~\ref{l:f.decay}, and
the residue theorem,
\[\begin{split}
\frac{\mathcal{W}_n(R)}{(2W)^n}
    &= N^{-1} \sum_{-N/2 < k \leq N/2} \big[ f(k/N) \big]^n \exp \frac{2\pi i Rk}{N} \\
    &= - \int_C \big[ f(z) \big]^n \frac{\exp \left[ 2\pi i R z \right]}
                                        {1 - \exp \left[ 2 \pi i N z \right]}  dz~,
\end{split}\]
where the contour $C$ encloses the zeros $k/N$, $-N/2 < k \leq N/2$, of the denominator,
as in Figure~\ref{fig:cont}.
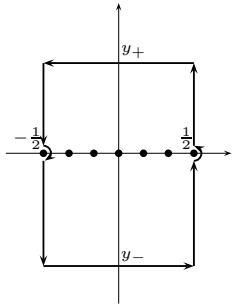
\begin{figure}[h]\label{fig:cont}
\vspace{1.7cm}
\setlength{\unitlength}{1cm}
\begin{pspicture}(-3,0)
\psset{linewidth=.3pt}
\psaxes[labels=none,ticks=none]{->}(0, 0)(-1.5, -2)(1.5, 2)
\psset{linewidth=.7pt,arrows=->}
\psline(-1, -1.5)(1, -1.5)
\psline(1, -1.5)(1, -.1)
\psarc(1, 0){.1}{270}{90}
\psline(1, .1)(1, 1.2)
\psline(1, 1.2)(-1, 1.2)
\psline(-1, 1.2)(-1, .1)
\psarcn(-1, 0){.1}{90}{280}
\psline(-1, -.1)(-1, -1.5)
\psset{dotsize=.1}
\psdot(0, 0)
\psdot(.33, 0)
\psdot(.66, 0)
\psdot(1, 0)
\psdot(-.33, 0)
\psdot(-.66, 0)
\psdot(-1, 0)
\rput(.2, 1.35){\tiny{$y_+$}}
\rput(.2, -1.38){\tiny{$y_-$}}
\rput(.9, .2){\tiny{$\frac{1}{2}$}}
\rput(-1.2, .2){\tiny{$-\frac{1}{2}$}}
\end{pspicture}
\vspace{-12mm}
\quad\quad\caption[Contour]{The contour.}\label{fig:diag3}
\vspace{16mm}
\end{figure}

\vspace{2mm}\noindent
The sum of integrals along the vertical parts of $C$ vanishes, hence
\[ \frac{\mathcal{W}_n(R)}{(2W)^n}
    = \int_{iy_+-1/2}^{iy_++1/2} -  \int_{iy_--1/2}^{iy_-+1/2}~.\]

\vspace{2mm}\noindent The first integral:
\[ I_+ = \int_{iy_+-1/2}^{iy_++1/2} =
    \int_{iy_+-1/2}^{iy_++1/2} \exp \phi_+(z) \frac{dz}{1 - \exp \left[ 2 \pi i N z \right]}~,\]
where $\phi_+(z) = n \ln f(z) + 2\pi i R z$. We choose $y_+$ to make $\phi_+'$ vanish at
$z_+ = i y_+$. Expanding $\phi_+$ in Taylor series at 0, we obtain:
\[ \phi_+(z) = - \frac{n\pi^2}{3} (W+1)(2W+1) z^2\,(1 + O(z^2 W^2)) + 2\pi i Rz~.\]
This expansion is of course differentiable, so
\[\begin{split}
\phi_+'(z) &= - \frac{2n\pi^2}{3} (W+1)(2W+1) z \,(1 + O(z^2 W^2)) + 2\pi i R~,\\
\phi_+''(z) &= - \frac{2n\pi^2}{3} (W+1)(2W+1) \,(1 + O(z^2 W^2))~,
\end{split}\]
therefore there exists a solution $z_+ = iy_+$ of $\phi_+'(z_+) = 0$ such that
\begin{equation}\label{eq:clt.sp1}\begin{split}
z_+ &= \frac{3iR}{\pi n (W+1)(2W+1)} \, \Big(1 + O(R^2/(n^2 W^2))\Big)~, \\
\phi_+(z_+) &= \frac{3R^2}{n (W+1)(2W+1)} + O(R^4/(n^3 W^4))~, \\
\phi_+''(z_+) &= - \frac{2n\pi^2}{3}(W+1)(2W+1) (1+O(R^2/(n^2 W^2)))~.
\end{split}\end{equation}
Under the assumptions of Proposition~\ref{p:clt1}, all the error terms in (\ref{eq:clt.sp1})
are $o(1)$. Also,
\[ \frac{1}{1 - \exp \left\{ 2\pi i N z \right\} }
    = 1 + O \left( \exp \left\{ - 2 \pi N y_+ \right\} \right) = 1 +o(1) \]
uniformly in $-1/2 \leq x \leq 1/2$, since
\[ N y_+  \geq \frac{cNR}{nW^2} \geq \frac{c N \eta}{\sqrt{n}W} \gg 1~.\]
These observations and Lemma~\ref{l:f.decay} justify the saddle point approximation, which
yields:
\[ I_+ = (1 + o(1)) \, \left[ \frac{\pi n}{3}(W+1) (2W+1) \right]^{-1/2}
    \exp \left\{ - \frac{3R^2}{n (W+1)(2W+1)} \right\}~.\]

\vspace{2mm}\noindent Now consider the second integral:
\[ I_- = -\int_{iy_--1/2}^{iy_-+1/2} =
    \int_{iy_--1/2}^{iy_-+1/2} \exp \phi_-(z)
        \frac{\exp \left[ 2\pi i N z \right]}{\exp \left[ 2 \pi i N z \right] - 1} dz~,\]
where $\phi_-(z) = n \ln f(z) + 2\pi i (R - N) z$. According to Lemma~\ref{l:f.decay},
\[\begin{split}
|I_-| &\leq \exp  \phi_-(iy_-) \, \int_{-1/2}^{1/2} \exp \left\{ - c n W^2 \right\} dx \\
      &\leq \frac{C}{\sqrt{n W}} \exp \phi_-(iy_-)~.
\end{split}\]
Choosing $y_-=-y_+$, we see that
\[ |I_-| \ll \frac{1}{\sqrt{n} W}
    \, \exp \left\{ - \frac{3R^2}{n (W+1)(2W+1)} \right\}~.\]
\end{proof}

\begin{proof}[Proof of Proposition~\ref{p:clt2}]
According to (\ref{eq:wn}),
\[ \frac{\mathcal{W}_n(R)}{(2W)^n}
    = N^{-1} \sum_{-N/2 < k \leq N/2} \big[ f(k/N) \big]^n \exp \frac{2\pi i Rk}{N}~. \]
We have: $f(0) = 1$, whereas by Lemma~\ref{l:f.decay}
\[ \left|  \big[ f(k/N) \big]^n \right| \leq \exp \left\{ - c n W^2 k^2 / N^2 \right\}~.\]
Therefore
\[ \left| \sum_{|k| > 0} \big[ f(k/N) \big]^n \exp \frac{2\pi i Rk}{N} \right|
    \leq C \exp \left\{ - \frac{cnW^2}{N^2} \right\} \ll 1~. \]
\end{proof}

\begin{proof}[Proof of Proposition~\ref{p:clt3} (sketch)]
Consider two cases: $\sqrt{n} W \geq N$, and $\sqrt{n}W < N$. In the former case,
\[ \frac{\mathcal{W}_n(R)}{(2W)^n}
    \leq \frac{1}{N} \sum_{-N/2 < k \leq N/2} \exp \left\{ - c n W^2 k^2 / N^2 \right\}
    \leq C/N~.  \]
In the latter case, we proceed as in the proof of Proposition~\ref{p:clt1}, making sure that
one can choose $y_\pm$ so that
\[ \phi_\pm(i y_\pm) \leq - \frac{cR^2}{nW^2}~. \]
\end{proof}

\section{Embeddings into $\G$}\label{s:emb}

In this section, we study the number of ``topological embeddings'' of a graph $G$
into $\G$, satisfying certain restrictions. Let us start with the precise definitions.\\*
Let $G = (V, E)$ be a connected multigraph, and let $\vv{n} = (n_1, \cdots, n_E)$
be an $E$-tuple of numbers\footnote{to simplify the notation, we write $V = \#V, E = \#E$
throughout this section}, $n_e \geq -1$. Construct a new graph $G^{\vv{n}}$ as follows:
\begin{itemize}
\item if $n_e \geq 0$, replace $e = (u, v)$ with a chain $(u, u_1, u_2, \cdots, u_{n_e}, v)$;
\item if $n_e = -1$, contract $e = (u, v)$ into a single vertex.
\end{itemize}
Finally, consider a system of linear equations $\Eq$ in the variables $n_1, \cdots, n_E$:
\[ \Eq: \quad \sum_{e \in E} c_j(e) n_e = n(j)~, \quad j = 1, \cdots, k~,\]
and denote by $\Emb(G, \Eq)$ the number of subgraphs $G' \hookrightarrow \G$ that are
isometric to $G^{\vv{n}}$ for some $\vv{n}$ satisfying the system $\Eq$, and by
$\Emb_+(G, \Eq)$ -- the number of such subgraphs with the additional requirement
$n_e > 0$, $e \in E$.
\\*
We impose the following restrictions on $\Eq$:
\begin{enumerate}
\item $\Eq$ is linearly independent;
\item $c_j(e) \geq 0$, $1 \leq j \leq k$, $e \in E$;
\item $\sum_{j=1}^k c_j(e) = 2$, $e \in E$.
\end{enumerate}
These restrictions are satisfied for the systems of equations that we need, and slightly
simplify the notation.

\begin{prop}\label{p:emb.RM}
In the setting above and under the assumptions 1.-3., if
$\frac{E^2 N^2}{W^2} \ll n(1) \leq n(2) \leq \cdots \leq n(k)$ and $\sum n(j) = 2n \ll W$, then
\[ \frac{\Emb(G, \Eq)}{(2W-1)^n}
    = (1 + o(1)) \, N^{-E+V} \, g_G(n(1), \cdots, n(k))~,\]
where $g_G$ is a continuous homogeneous function of degree $E-k$,
\[ g_G(n(1),\cdots,n(k)) \leq \frac{(Cn)^{E-k}}{(E-k)!}~.\]
The same is true for $\Emb_+(G, \Eq)$. If $\sum n(j) = 2n + 1$,
\[ \Emb(G, \Eq) = \Emb_+(G, \Eq) = 0~. \]
\end{prop}

\begin{prop}\label{p:emb.P}
In the setting above and under the assumptions 1.-3., if
$E^2 \ll n(1) \leq n(2) \leq \cdots \leq n(k) \ll \min \left( N^2/W^2, \, W \right)$,
$\sum n(j) = 2n$, then
\[ \frac{\Emb(G, \Eq)}{(2W-1)^n}
    = (1 + o(1)) \, N \, (2W)^{-E+V-1} h_G(n(1), \cdots, n(k))~,\]
where $h_G$ is a homogeneous function of degree $\frac{E + V - 2k - 1}{2}$, continuous
outside the origin, and
\[ h_G(n(1), \cdots, n(k)) \leq \frac{(Cn)^\frac{E+V-2k-1}{2}}{\left(\frac{E+V-2k-1}{2}\right){!}} ~.\]
The same is true for $\Emb_+(G, \Eq)$. If $\sum n(j) = 2n + 1$,
\[ \Emb(G, \Eq) = \Emb_+(G, \Eq) = 0~. \]
\end{prop}

\begin{prop}\label{p:emb.bd}
Without any restrictions on $n, W$,
\[ \frac{\Emb(G, \Eq)}{(2W)^n}
    \leq \frac{(C_1n)^{E-k}}{(E-k)!} (N)^{V-E}
    + C_2N \, \frac{(C_3n)^{\frac{E+V-2k-1}{2}}}{\left(\frac{E+V-2k-1}{2}\right)!}
        W^{-E+V-1}~.\]
\end{prop}

To apply the results of Section~\ref{s:rw}, we need a simple observation. Let
$u, v \in \ZZ / N \ZZ$, $A, B \subset \ZZ / N \ZZ$. Denote by
$\widetilde{\mathcal{W}}_n(u,v,A,B)$ the number of paths $p_n = u u_1 u_2 \cdots u_{n-1} v$
from $u$ to $v$ in $\G$, such that
\begin{itemize}
\item $u \neq u_2$, $u_1 \neq u_3$, \ldots, $u_{n-2} \neq v$ (``{\em non-backtracking}''),
\item $u_1 \notin A$, $u_{n-1} \notin B$.
\end{itemize}
\begin{lemma}\label{l:nbtrw} In the notation above,
\[ \mathcal{W}_n(|u-v|_N) \geq  \widetilde{\mathcal{W}}_n(u,v,A,B)
    \geq \left( 1 - C \, \frac{n + \# A + \# B}{W} \right) \, \mathcal{W}_n(|u-v|_N)~.\]
\end{lemma}

\begin{proof}[Proof of Proposition~\ref{p:emb.RM}]
Choose $N^2/W^2 \ll n_0 \ll n(1)/E^2$. Let $\Delta_\Eq$ be the set of real non-negative
solutions of $\Eq$; obviously, $\Delta_\Eq$ is an $(E-k)$-dimensional polytope with $E+k$ faces.
It is not hard to see that
\[ \Vol_{E-k-1} \partial \Delta_\Eq = O( \frac{E^2}{n(1)} \Vol \Delta_\Eq)~.\]
Hence, the number $\mathfrak{n}(\Eq)$ of integer solutions $\vv{n}$ of $\Eq$ with $n_e \geq -1$
satisfies:
\[ \mathfrak{n}(\Eq) = \gamma_\Eq \Vol_{E-k-1} \Delta_\Eq \left\{ 1 + O(E^2/n(1)) \right\}
    = \Vol \Delta_{\Eq}' \left\{ 1 + O(E^2/n(1)) \right\}~, \]
where $\Delta_{\Eq}' = \Proj_{\RR^{E-k-1}} \Delta_Eq$ is the projection of $\Delta_\Eq$ onto
a set of $E-k-1$ independent coordinates. Moreover, the number of solutions with at least
one coordinate $< n_0$ is at most
\begin{equation}\label{eq:nearbd}
\Vol \Delta_{\Eq}' \cdot O(n_0 E^2 / n(1))~.
\end{equation}
Let $\vv{n}$ be a solution with
\begin{equation}\label{eq:bigcoords}
n_1, \cdots, n_E \geq n_0~.
\end{equation}
It corresponds to
\[ N^{V-E} (2W-1)^n (1+o(1)) \]
different embeddings. Indeed, there are $N^V$ ways to embed the vertices of $G$, and,
according to Proposition~\ref{p:clt2} and Lemma~\ref{l:nbtrw},
\[ (2W-1)^n n^{-E} (1+o(1)) \]
ways to embed the edges. Thus the total number of embeddings corresponding to solutions
that satisfy (\ref{eq:bigcoords}) is
\begin{equation}\label{eq:tot}
\Vol \Delta_\Eq' \,  N^{V-E} (2W-1)^n (1+o(1))~.
\end{equation}
Applying (\ref{eq:nearbd}), Proposition~\ref{p:clt3} and Lemma~\ref{l:nbtrw}, we see
that the number of embeddings that violate (\ref{eq:bigcoords}) is negligible with respect
to (\ref{eq:tot}). Therefore
\[ \frac{\Emb(G, \Eq)}{(2W-1)^n} = (1 + o(1)) \, \Vol \Delta_\Eq' \,  N^{V-E}~. \]
Finally, observe that $\Vol \Delta_\Eq'$ is an $(E-k)$-homogeneous function of
$n(1)$, \ldots, $n(k)$, and that
\[ \Vol \Delta_\Eq'
    \leq \Vol \left\{ 0 \leq u_1, \cdots, u_E, \, u_1 + \cdots + u_E \leq n \right\}
    = \frac{n^{E-k}}{(E-k)!}~.\]
\end{proof}

\begin{proof}[Proof of Proposition~\ref{p:emb.P}] First, number the edges $1, \cdots, E$
so that $1, \cdots, V-1$ is a spanning tree. For $e = (i_e, t_e)$, let $R_e = i_e' - t_e'$,
where $i_e', t_e' \in \ZZ/N\ZZ$ are the images of $i_e, t_e$ in $G'$. Note that all the $R_e$
are linear combinations of $R_1, \cdots, R_{V-1}$. To every $e \in E$ we correspond a vector
$v_e \in \RR^{V-1}$ such that
\[ R_e = \sum_{f=1}^{V-1} v_e(f) R_f~.\]
By Proposition~\ref{p:clt1} and Lemma~\ref{l:nbtrw},
\begin{multline}\label{eq:emb.1}
\frac{\Emb(G, \Eq)}{(2W-1)^n} = (1+o(1)) \frac{\Emb(G, \Eq)}{(2W)^n}  \\
    = (1 + o(1)) \, \sum_{R_1, \cdots, R_{V-1}} {\sum_{\vv{n}}}^\ast \prod_{e \in E}
        \left[ \frac{2\pi}{3} n_e W^2 \right]^{-1/2}
        \exp \left\{ - \frac{3}{2} \frac{R_e^2}{n_e W^2} \right\}~,
\end{multline}
where the sum is over $n_e$ satisfying $\Eq$. Here we have disregarded the contribution
of $n_e \leq n_0$, where $1 \ll n_0 \ll n(1)$; this can be easily justified using
Proposition~\ref{p:clt3} and the asymptotic estimates below.
\\*
Now, applying the Poisson summation formula and discarding the negligible terms, we have:
\begin{multline*}
\sum_{R_1, \cdots, R_{V-1} \in \ZZ}
    \exp \left\{ - \frac{3}{2} \frac{R_e^2}{n_e W^2} \right\} \\
    = (1 + o(1)) \left[\frac{2\pi W^2}{3}\right]^\frac{V-1}{2}
       \left\{ \det  \left[ \sum_e n_e^{-1} v_e \otimes v_e\right] \right\}^{-1/2}~.
\end{multline*}
According to the Cauchy--Binet formula,
\[ \det \sum_e c_e v_e \otimes v_e
    = \sum_{T \in \binom{E}{V-1}} \prod_{e \in E} c_e \, {\det}^2 (v_e)_{e \in T}~.\]
If $T$ is not a spanning tree of $G$, $(v_e)_{e \in T}$ is not of full rank and hence has zero
determinant. If $T$ is a spanning tree, $(v_e)_{e \in T}$ and its inverse are integer matrices,
hence the squared determinant is equal to $1$, therefore
\[ \det \sum_e c_e v_e \otimes v_e = \sum_{T} \sum_{e \in T} c_e~,\]
where now the sum is over spanning trees $T \subset E$. Going back to (\ref{eq:emb.1}), we
deduce:
\begin{multline}\label{eq:emb.2}
\frac{\Emb(G, \Eq)}{(2W-1)^n} = (1+o(1)) \, \\
    \times \left[\frac{2\pi W^2}{3}\right]^\frac{-E+V-1}{2} \,
    {\sum_{\vv{n}}}^\ast
        \left( \prod_{e \in E} \frac{1}{\sqrt{n_e}} \right) \,
        \left( \sum_T \prod_{e \in T} \frac{1}{n_e} \right)^{-1/2}~.
\end{multline}
Finally, replace the sum $\sum^\ast$ with an integral over the polytope
\[ \Delta_{\Eq}' = \Proj_{\RR^{E-k}} \Delta_\Eq~,\]
where $\Delta_\Eq$ is the set of non-negative real solutions of $\Eq$, and the projection
is onto an independent subset of coordinates. This can be done for example using the Poisson
summation formula. We deduce:
\begin{multline}\label{eq:clt.sp2}
\frac{\Emb(G, \Eq)}{(2W-1)^n} = (1+o(1)) \, \left[\frac{2\pi W^2}{3}\right]^\frac{-E+V-1}{2} \, \\
     \times \idotsint\limits_{\Delta_\Eq'}
         \left( \prod_{e \in E} \frac{1}{\sqrt{u_e}} \right) \,
         \left( \sum_T \prod_{e \in T} \frac{1}{u_e} \right)^{-1/2} du_1 \cdots du_{E-k}
\end{multline}
(where all the other $u_e$ are determined from $\Eq$.) It is easy to see that the integral
in the right-hand side is $\frac{E+V-2k-1}{2}$--homogeneous in $n(1), \cdots, n(k)$.
Replacing $\Sigma_T$ with a single spanning tree $T_0$, we deduce the upper bound.
\end{proof}

\begin{ex}\label{ex:k1} The reader may find the following example illustrative: for $k = 1$,
\[ \Eq: \quad \sum 2 n_e = 2n~,\]
we have:
\[ \idotsint\limits_{\Delta_\Eq'} = n^{\frac{E+V-3}{2}} \idotsint\limits_{\widetilde\Delta_\Eq'}~,\]
where the set $\widetilde\Delta_\Eq'= n^{-1} \Delta_\Eq' = \{ \sum x_e \leq 1, x_e \geq 0\}$ does not actually depend of $n$.
\end{ex}

We omit the proof of Proposition~\ref{p:emb.bd}, which is very similar to the proofs of
the previous two propositions.

\section{Proofs of Propositions~\ref{p:main.rm}, \ref{p:main.rw}}\label{s:p.comb}

According to Lemma~\ref{l:fs} and the discussion preceding and following it,
\[ T(n(1), \cdots, n(k)) = \sum_{\mathpzc{D}} \mathrm{Paths}(n(1), \cdots, n(k); \mathpzc{D})~,\]
where the sum is over connected diagrams $\mathpzc{D} = (\bar{G} = (\bar{V}, \bar{E}), \bar{p})$, and
\[ \mathrm{Paths}(n(1), \cdots, n(k); \mathpzc{D}) \]
is the number of $k$-paths $p_{n(1), \cdots, n(k)}$ corresponding to $\mathpzc{D}$ and satisfying
(a)--(d) from Section~\ref{s:p}. Let
\[ c_j(\bar{e}) = \# \left\{ \begin{split}
    &\text{times that the $j$-th part}\\
    &\,\,\text{of $p$ passes through $\bar{e}$}
                       \end{split} \right\} \in \{0, 1, 2\}~,
    \quad 1 \leq j \leq k~, \, \bar{e} \in \bar{E}~,\]
and consider the system of equations
\[ \Eq_\mathpzc{D}: \quad \sum_{\bar{e}} c_j(\bar{e}) n_{\bar{e}} = n(j)~, \quad 1 \leq j \leq k~.\]
Then
\begin{equation}\label{eq:paths_emb}
\Emb_+(\bar{G}, \Eq_\mathpzc{D}) \leq \mathrm{Paths}(n(1), \cdots, n(k); \mathpzc{D})
    \leq \Emb(\bar{G}, \Eq_\mathpzc{D})~;
\end{equation}
here $\# \bar{V} = 2s$ and $\# \bar{E} = 3s-k$. Therefore we can apply the results of
Section~\ref{s:emb}.

\begin{proof}[Proof of Proposition~\ref{p:main.rm}]
Choose $n_0$ and $s_0$ so that
\[ N^{1/3} \gg n_0 \gg \frac{s_0 N^2}{W^2} \gg \frac{N^2}{W^2}~.\]
If $n_0 \leq n(1) \leq \cdots \leq n(k) \leq R N^{1/3}$, and $\mathpzc{D}$ is a
connected diagram with $2s$ vertices and $3s-k$ edges, $1 \leq s \leq s_0$, then
Proposition~\ref{p:emb.RM} yields:
\[\begin{split}
&\frac{\mathrm{Paths}(n(1), \cdots, n(k); \mathpzc{D})}{(2W-1)^n} \\
&\qquad= (1+o(1)) \, g_\mathpzc{D}(n(1), \cdots, n(k)) N^{-s+k} \\
&\qquad= (1+o(1)) \, n(1) \cdots n(k) \,
    \widetilde{g}_\mathpzc{D}(n(1)/N^{1/3}, \cdots, n(k)/N^{1/3})~,
\end{split}\]
where $g_\mathpzc{D}$ is homogeneous of degree $3s-k$, and
\[ \widetilde{g}_\mathpzc{D}(x_1, \cdots, x_k)
    = (x_1 \cdots x_k)^{-1} g_\mathpzc{D}(x_1, \cdots, x_k)~.\]
If $s > s_0$, Proposition~\ref{p:emb.bd} yields:
\[ \frac{\mathrm{Paths}(n(1), \cdots, n(k); \mathpzc{D})}{(2W-1)^n}
    \leq \frac{(C'n)^{3s-2k}}{(3s-2k)!} \, (C'N)^{-s+k}~,\]
hence the total contribution of all diagrams with $s > s_0$ is at most
\[\begin{split}
&C_k \sum_{s>s_0} \frac{(C'n)^{3s-2k}}{(3s-2k)!} \, (C'N)^{-s+k} \, (C s)^{s+k-1}\\
&\qquad\leq C_1(k) n^k \sum_{s > s_0} \frac{(C_1 n^3/N)^{s-k}}{s^{2s-3k+1}}
    = o(n^k) = o(N^{k/3})~.
\end{split} \]
As the same bound holds for the sum of
\[ n(1) \cdots n(k) \, \widetilde{g}_\mathpzc{D}(n(1)/N^{1/3}, \cdots, n(k)/N^{1/3})\]
over diagrams with $s>s_0$, we deduce:
\begin{multline*}
\frac{T(n(1), \cdots, n(k))}{(2W-1)^n} \\
    = \left[ \prod n(j) \right] \,
    \sum_\mathpzc{D} \widetilde{g}_\mathpzc{D}(n(1)/N^{1/3}, \cdots, n(k)/N^{1/3})
    + o(N^{k/3})~,
\end{multline*}
where the sum is now over all connected $k$-diagrams $\mathpzc{D}$. To prove the bound for
$n(1) < n_0$, apply Proposition~\ref{p:emb.RM} in a similar fashion.
\end{proof}

\begin{proof}[Proof of Proposition~\ref{p:main.rw} (Sketch)]
Choose $s_0, n_0, n_0'$ so that
\[ 1 \ll s_0^2 \ll n_0 \ll W^{2/5} \ll n_0' \ll \frac{N^2}{W^2}~, \]
and proceed as in the previous proof, using Proposition~\ref{p:emb.P} instead of
Proposition~\ref{p:emb.RM}.
\end{proof}

\section{Digression}\label{s:sl}

Let $\{ \mu_N \}_{N=1}^\infty$ be a sequence of probability measures on $\RR$,
and let $s_N \to +\infty$ be a sequence of (real) positive numbers. We shall study
the scaled measures
\begin{equation}\label{eq:sc}
\begin{cases}
    &\sigma_{R,N}(\lambda) = s_N^3 \left( 1 - \mu_N( 1 - 2 s_N^{-2} \lambda ) \right)\\
    &\sigma_{L,N}(\lambda) = s_N^3 \left( \mu_N( -1 + 2 s_N^{-2} \lambda) \right)
\end{cases}~, \quad \lambda \in \RR~.
\end{equation}
This scaling is meaningful if $\mu_N$ are close (in some sense) to the Wigner
measure $\mu_{\operatorname{Wigner}}$.
\\*
Let
\begin{equation}\label{eq:muhat}
\widehat{\mu_N}(n) = \int U_n(\alpha) d\mu_N(\alpha)~,
\end{equation}
where as before $U_n$ are the Chebyshev polynomials of the second kind, and
assume that
\begin{equation}\label{eq:hat.asymp}
\widehat{\mu_N}(n) = \frac{n}{s_N^3} \Big[ \phi_R(n/s_N) + (-1)^n \phi_L(n/s_N) \Big]
    + \frac{\eps_N(n)}{s_N^2}~,
\end{equation}
where
\begin{equation}\label{eq:phiprop}
\phi_L, \phi_R \in C(0, +\infty) \bigcap L_1(\exp \left\{ - x^{-2+\delta} \right\} dx)
\end{equation}
(for some $\delta > 0$), and $\eps_N(n)$ are ``small'' (in a sense made precise below),
for $n = O(s_N)$. Then $\sigma_{R,N}, \sigma_{L,N}$ converge to limits that can be expressed
in terms of $\phi_R, \phi_L$. We state this as a proposition; the main ingredient of the
proof is a variant of Levitan's uniqueness theorem \cite{L}. The assumptions can be
definitely relaxed: thus, the integrability condition (\ref{eq:phiprop}) can be replaced
with a weaker one using the methods of Levitan and Meiman \cite{LM} and Vul \cite{V}.

\begin{prop}\label{p:6}
Let $\{\mu_N\}$ be a sequence of measures on $\RR$, and assume that the coefficients
$\widehat{\mu_N}$ (defined in (\ref{eq:muhat})) satisfy (\ref{eq:hat.asymp}), with
$\phi_L$, $\phi_R$ as in (\ref{eq:phiprop}), and
\begin{enumerate}
\item $\sum_{n=1}^{s_N} \frac{|\eps_N(n)|}{n} \leq C$;
\item for any $R \in \NN$, $\eps_N(n) = o(1)$ for $n \in \{1,2,\cdots,Rs_N\}$, with the
implicit constant depending only on $R$.
\end{enumerate}
Then the measures $\sigma_{R,N}$ and $\sigma_{L,N}$ (defined in (\ref{eq:sc})) converge to
limiting measures $\sigma_R$ and $\sigma_L$ (respectively), where $\sigma_R$ and
$\sigma_L$ are uniquely defined by $\phi_R$ and $\phi_L$ (respectively). The limiting
measures $\sigma = \sigma_R, \sigma_L$ share the following properties:
\begin{equation}\label{eq:limprop}
\begin{cases}
\left|\sigma(\lambda) - \frac{2}{3\pi} \lambda_+^{3/2} \right| = O(\lambda)~, &\lambda \to +\infty~; \\
\sigma(\lambda) = O \left[ \exp \left\{ - C' |\lambda|^{\frac{1 - \delta/2}{1 - \delta}}\right\} \right]~,
    &\lambda \to -\infty~.
\end{cases}
\end{equation}
\end{prop}

\begin{rmk}
One can show that
\begin{equation}\label{eq:gl}
\begin{split}
\int_{-\infty}^{+\infty} \frac{\sin x\sqrt{\lambda}}{x \sqrt \lambda} d\tau_R(\lambda)
    &= \phi_R(x) \\
\int_{-\infty}^{+\infty} \frac{\sin x\sqrt{\lambda}}{x \sqrt \lambda} d\tau_L(\lambda)
    &= \phi_L(x)
\end{split}~,\quad
\begin{split}
\sigma_R(\lambda) = \tau_R(\lambda) + \frac{2}{3\pi} \lambda_+^{3/2}\\
\sigma_L(\lambda) = \tau_L(\lambda) + \frac{2}{3\pi} \lambda_+^{3/2}
\end{split}~.\end{equation}
We do not use this in the sequel, and therefore omit the proof.
\end{rmk}

To prove Proposition~\ref{p:6}, we shall need the following Erd\H{o}s--Tur\'an type
inequality:
\begin{prop}[{\cite[Proposition~5]{FS2}}]\label{p:et}
Let $\mu$ be a probability measure on $\RR$. Then, for any $s \geq 1$
and any $\alpha \in \RR$,
\[ \big| \mu(\alpha) - \mu_{\operatorname{Wigner}}(\alpha) \big| \\
    \leq C \left\{ \frac{\rho(\alpha; s)}{s}
        + \sqrt{\rho(\alpha; s)}\,\sum_{n=1}^{s} \frac{|\widehat{\mu}(n)|}{n} \right\}~, \]
where $\rho(\alpha; s) = \max(1 - |\alpha|, s^{-2})$.
\end{prop}

\begin{proof}[Proof of Proposition~\ref{p:6}] \hfill\\

\newcounter{remnr3}
\def\rem3{
    \addtocounter{remnr3}{1}
    \vspace{2mm}\noindent{\bf (\alph{remnr3})} }
\rem3 According to assumption~1.\ and Proposition~\ref{p:et},
\begin{equation}\label{eq:etcor}
|\tau_{R,N}(\lambda)| \leq C \max(\lambda, 1)~,
\end{equation}
where
\[ \tau_{R,N}(\lambda) = \sigma_{R,N}(\lambda)
    - \frac{2}{\pi} \int_{1-2s_N^{-2} \lambda_+}^1 \sqrt{1-\alpha^2} \, d\alpha~. \]
Therefore the sequence $\{ \sigma_{R,N} \}$ is precompact. The same is of
course true for $\{ \sigma_{L,N} \}$. If $\sigma_{R, N_j} \to \sigma_R$,
then $\tau_{R, N_j} \to \tau_R$, defined by
\[ \sigma_R(\lambda) = \tau_R(\lambda) + \frac{2}{3\pi} \lambda_+^{3/2}~.\]
Therefore we need to prove that $\sigma_R$, $\tau_R$ are uniquely determined.

\rem3 Let
\begin{equation}\label{eq:jk}
\begin{split}
U_n^4(\alpha) &= \sum_{k=0}^{4n} c_{n,k}^+ U_k(\alpha)~,\\
U_{n+1}(\alpha) U_n^3(\alpha) &= \sum_{k=0}^{4n+1} c_{n,k}^- U_k(\alpha)~.
\end{split}
\end{equation}
The explicit expressions for $c_{n,k}^\pm$ can be easily derived from the identity
\[ U_k(\alpha) U_\ell(\alpha) = \sum_{m=0}^{\min(k, \ell)} U_{|\ell - k| + 2m}(\alpha)~.\]
We shall only need the following simple properties:
\[\begin{split}
c_{n,k}^+ &= 2 \1_{2\ZZ}(k) \, (c(k/n) + o(1)) \, n^2~, \\
c_{n, 0}^+ &= (c_0 + o(1)) \, n~,\\
c_{n,k}^- &= 2 \1_{2\ZZ+1}(k) \, (c(k/n) + o(1)) \, n^2~,
\end{split}\]
where $c \in C[0, 4.01]$ and the $o(1)$ terms are uniform. Thus we have:
\[\begin{split}
&\frac{s_N^3}{n^4} \int U_n^4(\alpha) d\mu_N(\alpha) \\
    &\quad= \int_0^4 x c(x) \left[ \phi_R(xn/s_N) + \phi_L(xn/s_N) \right] dx
        + c_0 s_N^3/n^3 + o(1)~, \\
&\frac{s_N^3}{n^4} \int U_{n+1}(\alpha) U_n^3(\alpha) d\mu_N(\alpha) \\
    &\quad= \int_0^4 x c(x) \left[ \phi_R(xn/s_N) - \phi_L(xn/s_N) \right] dx
        + o(1)~.
\end{split}\]
Passing to the limit along a subsequence $N_j \to \infty$ and applying
(\ref{eq:etcor}) and the dominated convergence theorem, we deduce:
\begin{equation}\label{eq:gl4}
\begin{split}
\int \frac{\sin^4 x\sqrt{\lambda}}{(x\sqrt{\lambda})^4} d\sigma_R(\lambda)
    &= \int_0^4 y c(y) \phi_R(xy) dy + c_0 x^{-3}~,\\
\int \frac{\sin^4 x\sqrt{\lambda}}{(x\sqrt{\lambda})^4} d\sigma_L(\lambda)
    &= \int_0^4 y c(y) \phi_L(xy) dy + c_0 x^{-3}~.\\
\end{split}\end{equation}

\rem3 The relations (\ref{eq:gl4}) determine $\sigma_R, \sigma_L$ uniquely.
Indeed, according to (\ref{eq:etcor}) and (\ref{eq:phiprop}),
the measure $\sigma_R$ must satisfy
\begin{equation}\label{eq:decay}
\int_0^\infty \frac{d\sigma_R(\lambda)}{1 + \lambda^2} < +\infty~,
\quad \int_{-\infty}^0 \exp \left[ x\sqrt{|\lambda|} \right] \, d\sigma_R(\lambda)
    \leq C_1 \exp \left[ C_1 x^{2-\delta}\right]~;
\end{equation}
hence one may apply the argument of Levitan \cite{L}. The latter works
as follows. Suppose $\widetilde\sigma_R$ is another measure for which (\ref{eq:gl4})
also holds. Then $\widetilde\sigma_R$ also satisfies (\ref{eq:decay});
$\nu = \sigma_R - \widetilde\sigma_R$ satisfies
\[ \int \frac{\sin^4 x \sqrt{\lambda}}{(x\sqrt{\lambda})^4} d\sigma(\lambda) = 0~,
    \quad x \in \RR~.\]
Set
\[ f(x) = \int_{-\infty}^0 \frac{\sin^4 x \sqrt{\lambda}}{\lambda^2} d\nu(\lambda)
        = - \int_0^{+\infty} \frac{\sin^4 x \sqrt{\lambda}}{\lambda^2} d\nu(\lambda)~,\]
and apply the Phragm\'en--Lindel\"of principle to $f(z)$ in every quadrant. We see
that $f$ is bounded, hence constant. Therefore $\widetilde\sigma_R = \sigma_R$.

We have proved that the limiting measure is unique; in particular, the sequence
$\{\sigma_R,N\}$ converges (to this limit.)

\rem3 Returning to (\ref{eq:decay}),
\[ \int_{-\infty}^\infty \exp \left\{ x \sqrt{|\lambda|} \right\} d\sigma_R(\lambda)
\leq C_1 \exp \left[ C_1 x^{2-\delta} \right]~,\]
hence
\[ \sigma_R(\lambda) \leq C_1 \exp \left[ - x \sqrt{|\lambda|} + C_1 x^{2-\delta} \right]~,
    \quad \lambda \leq 0~, x \geq 0~. \]
Taking $x = \left[ \frac{\sqrt{|\lambda|}}{2C_1}\right]^\frac{1}{1-\delta}$, we obtain
the second part of (\ref{eq:limprop}). The first part follows from (\ref{eq:etcor}).
\end{proof}

Proposition~\ref{p:6} can be extended to measures on $\RR^\ell$ (for any {\em fixed} $\ell$.)
Namely, let $\{\mu_N\}$ be a sequence of measures on $\RR^\ell$. For simplicity, we assume that
$\mu_N$ are symmetric (=invariant under permutation of coordinates). For $k \leq \ell$, set
\[ \widehat{\mu_N} (n(1), \cdots, n(k)) = \int \prod_{j=1}^k U_{n(j)} (\alpha_j) d\mu_N(\alpha)~.\]

\begin{prop}\label{p:6.k}
Assume that, for $1 \leq k \leq \ell$,
\begin{multline*}
\widehat{\mu_N}(n(1), \cdots, n(k)) \\
    = \frac{\prod_{j=1}^k n(j)}{s_N^{3k}}
        \sum_{I \subset \{1, \cdots, k \}} (-1)^{\sum_{j \in I} n(j)}
        \phi_{R, k}\left( \{n(j)\}_{j \notin I} \right)
        \phi_{L, k}\left( \{n(j)\}_{j \in I} \right) \\
    + \frac{\eps_N(n(1), \cdots, n(k))}{s_N^{2k}}~,
\end{multline*}
where
\[ \phi_{R, k}, \phi_{L, k} \in C\left((0, +\infty)^k\right)
    \bigcap L_1(\exp \left[ - \|x \|^{2-\delta} \right] dx)~, \]
the coefficients $\eps_N$ tend to zero uniformly on $\vv{n} \in \{1, \cdots, R s_N \}^k$,
and
\[ \sum_{1 \leq n(1) \leq \cdots \leq n(k) \leq s_N}
    \frac{\eps_N(n(1), \cdots, n(k))}{n(1) \cdots n(k)} \leq C~.\]
Then the scaled measures
\begin{equation}\label{eq:sc.k}
\begin{cases}
    &\sigma_{R,N}(\lambda_1, \cdots, \lambda_\ell)
        = s_N^3 \left( 1 - \mu_N( 1 - 2 s_N^{-2} \lambda_1, \cdots, 1 - 2s_N^{-2} \lambda_\ell) \right) \\
    &\sigma_{L,N}(\lambda_1, \cdots, \lambda_\ell)
        = s_N^3\, \mu_N( -1 + 2 s_N^{-2}\lambda_1, \cdots, -1+2s_N^{-2} \lambda_\ell)
\end{cases}
\end{equation}
converge to limiting measures $\sigma_R, \sigma_L$, which are uniquely determined by
$\{\phi_{k,R}\}$, $\{\phi_{k, L}\}$ (respectively). Moreover, $\sigma = \sigma_L, \sigma_R$
satisfy:
\begin{equation*}
\begin{cases}
\left|\sigma(\lambda_1, \cdots, \lambda_\ell)
    - \left(\frac{2}{3\pi}\right)^\ell
        \prod_{j=1}^\ell \lambda_{j \,+}^{3/2} \right| = O(\|\lambda\|^\frac{3k-1}{2})~,
            &\lambda_j \geq 0~, \, \|\lambda \| \to \infty; \\
\sigma(\lambda_1, \cdots, \lambda_\ell) = O \left[ \exp \left\{ - C'_\ell \|\lambda\|^{\frac{1 - \delta/2}{1 - \delta}}\right\} \right]~,
    &\lambda_j \leq 0~, \, \|\lambda \| \to \infty~.
\end{cases}
\end{equation*}
\end{prop}

The proof is similar to the one-dimensional case (Proposition~\ref{p:6}); we omit it.

\section{Proof of the main results}\label{s:pr}

\begin{proof}[Proof of Theorem~\ref{th:RM}]
To show that the random counting measures $\sigma_R(\lambda)$, $\sigma_L(\lambda)$
converge in distribution to $\mathfrak{Ai}_1(-\lambda)$, it is sufficient to
prove the convergence of the correlation measures
\[ \rho_{\ell, R}(\lambda_1, \cdots, \lambda_\ell) = \EE \prod_{j=1}^\ell \sigma_R(\lambda_j),
\quad \rho_{\ell, L}(\lambda_1, \cdots, \lambda_\ell) = \EE \prod_{j=1}^\ell \sigma_L(\lambda_j) \]
to
\[ \rho_\ell(\lambda_1, \cdots, \lambda_\ell) =  \EE \prod_{j=1}^\ell \mathfrak{Ai}_1(-\lambda_j)~.\]
It will be convenient to denote
\begin{equation}\label{eq:def.mu.1}
\mu_N(\alpha) = \# \left\{ \text{eigenvalues of $\frac{H_N}{2\sqrt{2W_N-1}}$
    in $(-\infty, \alpha]$} \right\}
\end{equation}
(this differs slightly from (\ref{eq:def.mu}).)

Let us first consider $\ell = 1$. According to Proposition~\ref{p:main.rm},
\[\begin{split}
\widehat{\EE \mu_N}(2n)
    &= \frac{n}{N} \phi_1(n/N^{1/3}) + \frac{\eps_N^{(1)}(n)}{N^{2/3}} \\
    &= \frac{2n}{N} \left[ \widetilde{\phi}_1(2n/N^{1/3})
        + (-1)^{2n} \widetilde{\phi}_1(2n/N^{1/3})\right]
        + \frac{\eps_N^{(2)}(2n)}{N^{2/3}}~, \\
\widehat{\EE \mu_N}(2n)
    &= 0\\
    &= \frac{2n+1}{N} \left[ \widetilde{\phi}_1((2n+1)/N^{1/3})
        + (-1)^{2n+1} \widetilde{\phi}_1((2n+1)/N^{1/3})\right]~, \\
\end{split}\]
where $\eps_N^{(1)}$ absorb the difference between the matrices
${H_N^{(2n)}} / (2W_N-1)^n$
and $U_{2n}(H_N/(2\sqrt{2W_N})$, and $\eps_N^{(2)}$, $\widetilde{\phi}_1$ are introduced
to make the notation compatible with Section~\ref{s:sl}. The sequence of measures $\{\mu_N\}$
satisfies the assumptions of Proposition~\ref{p:6} with $s_N = N^{1/3}$, hence
\[ \rho_{1,R},\rho_{1,L} \to \widetilde{\rho}_1~,\]
where $\widetilde{\rho}_1$ is a measure determined by $\widetilde{\phi}_1$ (and in particular
independent of $W_N$.) For $W_N = N/2$, $\widetilde{\rho}_1 = \rho_1$ according to the result
of Soshnikov \cite{S1}; hence the same is true for any $W_N \gg N^{5/6}$.

The same argument works for $\ell > 1$. Indeed,
\[\begin{split}
\EE \prod_{j=1}^\ell \tr \frac{H_N^{n(j)}}{(2W_N-1)^n}
    &= \sum_\Pi \prod_{P \in \Pi} T(\{n(j)\}_{j \in P} \\
    &= \sum_\Pi \prod_{P \in \Pi} \frac{1+(-1)^{n(j)}}{2} T(\{n(j)\}_{j \in P}~,
\end{split}\]
where the sum is over all partitions $\Pi$ of $\{1, \cdots, \ell\}$. For a subset
$I \subset \{1, \cdots, \ell \}$, write $\Pi \prec I$ if
\[ \forall P \in \Pi \,\, P \cap I \in \{ P, \varnothing \}~.\]
Then
\[ \EE \prod_{j=1}^\ell \tr \frac{H_N^{n(j)}}{(2W_N-1)^n}
    = \sum_{I \subset \{1, \cdots, \ell \}} (-1)^{\sum_{j \in I} n(j)}
        \prod_{\Pi \prec I} \frac{ T( \{n(j)\}_{j \in P} )}{2}~.\]
Now apply Proposition~\ref{p:main.rm} and then Proposition~\ref{p:6.k}.
\end{proof}

\begin{rmk} Another (perhaps, slightly simpler) way to prove the convergence
of the correlation measures is to follow the arguments of \cite[Section~I.5]{FS},
and then use the uniqueness theorem for Laplace transform instead of the arguments of
Section~\ref{s:sl}, as in \cite{S1}.
\end{rmk}

\begin{proof}[Proof of Theorem~\ref{th:P}]
According to Proposition~\ref{p:main.rw} with $k=1$,
\[ \widehat{\EE \mu_N} (2n) = \frac{1}{W_N^{6/5}} \, 2n \cdot 2 \psi_1^{(1)}(2n/W_N^{2/5})
    + \frac{\eps_N^{(1)}}{W_N^{4/5}}~,\]
where again $\psi_1^{(1)}$, $\eps_N^{(1)}$ are introduced to make the notation consistent
with Section~\ref{s:sl}. Now apply Proposition~\ref{p:6} with $s_N = W_N^{2/5}$,
and deduce that
\[ \EE \sigma_{R,N} , \, \EE \sigma_{L,N} \To \sigma_1~,\]
where
\[ \sigma_1(\lambda) = \tau_1(\lambda) + \frac{2}{3\pi} \lambda_+^{3/2}~, \quad
    \int_{-\infty}^{+\infty} \frac{\sin x \sqrt{\lambda}}{ x \sqrt \lambda}
        \, d\tau_1(\lambda) = \psi(x)~.
\]
Applying Proposition~\ref{p:main.rw} with $k=2$ and Proposition~\ref{p:6.k},
it is not hard to see that
\[\begin{split}
\EE \sigma_{R,N} \otimes \sigma_{R,N} &- (\EE \sigma_{R,N}) \otimes (\EE \sigma_{R,N}) \To 0~, \\
\EE \sigma_{L,N} \otimes \sigma_{L,N} &- (\EE \sigma_{L,N}) \otimes (\EE \sigma_{L,N}) \To 0~,
\end{split}\]
hence also
\[ \sigma_{R,N}, \sigma_{L,N} \ToD \sigma_1~.\]
\end{proof}

\begin{rmk}\label{rmk:conj} Staring at the asymptotics of $\psi$ near zero, it
seems natural to conjecture that
\begin{equation}\label{eq:as.sigma}
\sigma_1(\lambda) = \frac{2}{3\pi} \lambda_+^{3/2} + \sqrt\frac{3}{32 \pi^2} \, \lambda_+^{1/4}
    + O(1)~, \quad \lambda \to + \infty~.
\end{equation}
We have not been able to prove this as stated. Applying Marchenko's Tauberian theorem \cite{Ma},
one can show that (\ref{eq:as.sigma}) holds in a weak sense (say, after integrating both
sides with a compactly supported twice differentiable kernel.)
\end{rmk}

\begin{proof}[Proof of Theorem~\ref{th:norm}]
First, for $n \leq W_N$, Proposition~\ref{p:emb.bd} yields:
\[\begin{split}
\frac{\EE \tr H_N^{(2n)}}{(2W_N-1)^N}
    &= \sum_\mathpzc{D} \frac{\mathrm{Paths}(2n; \mathpzc{D})}{(2W_N-1)^n} \\
    &= \sum_\mathpzc{D} \left[ \frac{(Cn)^{3s-2}}{(3s-2)!} (cN)^{-s+1}
        + N \, \frac{(Cn)^\frac{5s-4}{2}}{ \left(\frac{5s-4}{2}\right) !} (cW_N)^{-s} \right]~.
\end{split} \]
Rearranging the sum and using Lemma~\ref{l:fs}, we continue:
\[\begin{split}
\frac{\EE \tr H_N^{(2n)}}{(2W_N-1)^N}
    &\leq \sum_s \left[ \frac{(Cn)^{3s-2}}{(3s-2)!} (cN)^{-s+1}
        + N \, \frac{(Cn)^\frac{5s-4}{2}}{ \left(\frac{5s-4}{2}\right) !} (cW_N)^{-s} \right]
        \, (Cs)^s \\
    &\leq Cn \left\{ \exp \left[ \frac{Cn^{3/2}}{N^{1/2}}\right]
        + \frac{N}{W_N^{6/5}} \exp \left[ \frac{Cn^{5/3}}{W_N^{2/3}} \right] \right\}~,
\end{split} \]
hence also
\[ \EE \tr U_{2n} \left( \frac{H_N}{2\sqrt{2W_n}} \right)\
    \leq Cn \left\{ \exp \left[ \frac{Cn^{3/2}}{N^{1/2}}\right]
        + \frac{N}{W_N^{6/5}} \exp \left[ \frac{Cn^{5/3}}{W_N^{2/3}} \right] \right\} \]
(perhaps, with a different constant $C$.) Applying the identities (\ref{eq:jk}),
we have:
\[ \EE \tr U_{2n}^4 \left( \frac{H_N}{2\sqrt{2W_n}} \right)\
    \leq C \left\{ n N + n^4 \exp \left[ \frac{Cn^{3/2}}{N^{1/2}}\right]
        + \frac{N n^4 }{W_N^{6/5}} \exp \left[ \frac{Cn^{5/3}}{W_N^{2/3}} \right]
    \right\}~. \]
If $N^{5/6} \leq W_N$, the right hand side is bounded by $Cn^4$ for
$n = \lfloor N^{1/3} \rfloor$. As
\[ \frac{U_{2n}(\pm(1+\eps))}{2n} \geq c \exp(c n \sqrt{\eps})~,  \]
we deduce that $\|H_N/(2\sqrt{2W_N})\| \ToD 1$ (and in fact,
\[ \left\{ \left( \|H_N/(2\sqrt{2W_N}) - 1 \right) N^{2/3} \right\}_N \]
is stochastically bounded.)
\\*
If $W_N \leq N^{5/6}$, take $n = \lfloor W_N^{3/5} \log^{2/5} N \rfloor$.
Then
\[ \EE \tr U_{2n}^4 \left( \frac{H_N}{2\sqrt{2W_n}} \right)
    \leq C n^4 N^{C'} \big/ W_N^{6/5}~,\]
and hence again $\|H_N/(2\sqrt{2W_N})\| \ToD 1$.
\end{proof}

\section{Random phases}\label{s:ph}

The steps of the proof for the matrices with entries (\ref{eq:ph}) are very similar
to those for (\ref{eq:pm}). We indicate the necessary modifications.

\begin{itemize}
\item Section~\ref{s:p}: Lemma~\ref{l:fs} remains valid {\em verbatim}. Condition (d)
in Corollary~\ref{cor} should be replaced with
\[ \# \left\{ (i, j) \, | \, u_i^j = u, \, u_{i+1}^j = v \right\}
    = \# \left\{  (i, j) \, | \, u_i^j = v, \, u_{i+1}^j = u  \right\}~.\]
In Definition~\ref{def:diag.k}, loops are not allowed, and the third condition
should be also replaced with
\[ \# \left\{ (i, j) \, | \, \bar{u}_i^j = \bar{u}, \, \bar{u}_{i+1}^j = \bar{v} \right\}
    = \# \left\{  (i, j) \, | \, \bar{u}_i^j = \bar{v},
        \, \bar{u}_{i+1}^j = \bar{u} \right\} = 1~.\]
Thus, the diagrams in the new sense are a subset of diagrams in the old sense.
In Lemma~\ref{l:fs}, $s$ is now always even, and the estimate is valid for
even $s$. In Propositions~\ref{p:main.rm},\ref{p:main.rw} the sums are now
over even $s$ (and the functions $\phi, \psi$ are different than before.)
\item In Section~\ref{s:p.comb}, one should only consider the diagrams
that are valid according to the new definition, and only even values of $s$.
\item The argument in Section~\ref{s:pr} is still valid. In Remark~\ref{rmk:conj},
we would now conjecture that
\begin{equation}\label{eq:as.sigma.2}
\sigma_2(\lambda) = \frac{2}{3\pi} \lambda_+^{3/2} + O(1)~,
    \quad \lambda \to + \infty~.
\end{equation}
\end{itemize}

\section{Concluding remarks}\label{s:cr}

\newcounter{remnr4}
\def\rem4{
    \addtocounter{remnr4}{1}
    \vspace{2mm}\noindent{\bf \Roman{remnr4}.} }

\rem4 To simplify the exposition, we have only considered the simplest random
variables (\ref{eq:pm}), (\ref{eq:ph}). Assume that the entries of $H_N$ above the
diagonal are independent, and have symmetric distribution with (uniformly) subgaussian
tails. Applying the methods of \cite[Part III]{FS}, one can assume that
\[ \EE H_{uv}^2 = 1~, \quad H_{uv} \in \RR \,\,\, \mathrm{a.s.} \quad (0 < |u-v|_N \leq W)\]
or
\[ \EE H_{uv}^2 = 0~, \quad \EE |H_{uv}|^2 = 1 \quad (0 < |u-v|_N \leq W)\]
instead of (\ref{eq:pm}) or (\ref{eq:ph}) (respectively), and prove analogues of
Theorems~\ref{th:RM} and \ref{th:P}. In particular, this extension covers the frequently
considered case of matrices with Gaussian elements.

\rem4 The restriction $W_N \gg 1$ is also an artefact of the proof. For $W_N = O(1)$,
Lemma~\ref{l:nbtrw} is no longer applicable, hence one should take into account the
difference between non-backtracking and usual random walk. Thus, we need an analogue
of Propositions~\ref{p:clt1}--\ref{p:clt3} for non-backtracking walks. This can probably
be proved using either a trace formula for the representation of non-backtracking
random walk as a Markov chain on the space of directed edges (see Smilansky \cite{Sm}),
or the connection to Chebyshev polynomials (see \cite{ABLS}.)

\rem4 It would be interesting to obtain a more detailed description of the measure
$\sigma_\beta$ from Theorem~\ref{th:P}. In particular, a more precise description of the
left tail would allow to find the limiting distribution of the maximal eigenvalue (cf.\ the
remark after Theorem~\ref{th:norm}); as to the right tail, it would be interesting
to justify the asymptotics (\ref{eq:as.sigma}) (and perhaps derive the next terms
in the asymptotic series.)

\rem4 Usual (i.e.\ non-periodic) random band matrices have non-zero elements $H_{uv}$
for
\[ 0 < |u - v| \leq W_N~.\]
We expect that the results of this paper hold, perhaps in modified form, for these
matrices as well. Following Bogachev, Molchanov, and Pastur \cite{BMP}, we note however
that even the limiting spectral measure coincides (\ref{eq:w}) only if $1 \ll W_N \ll N$
or $W_N = (1-o(1)) N$. The limiting spectral measure in the complementary regimes
has been described by Khorunzhiy, Molchanov, and Pastur in \cite{MPK}.

\rem4 It would also be interesting to study the crossover regime $W_N \asymp N^{5/6}$.
We refer the reader to the works of Johansson \cite{J} and Bender \cite{Be}
for the description of the crossover regime at the spectral edge for different kinds of
random matrices.

\rem4 The method of this paper can be used to study the eigenvectors of
$H_N/(2\sqrt{2W_N})$ that correspond to eigenvalues close to the edge, and, in particular,
their {\em inverse participation ratio}
\[ \sum_{u=1}^N |\mathrm{v}(u)|^4 \big/ \left( \sum_{u = 1}^N |\mathrm{v}(u)|^2 \right)^2~.\]
If $W_N \gg N^{5/6}$, the inverse participation ratio of eigenvectors corresponding to
eigenvalues $\alpha = 1 + O(N^{-2/3})$ is, with high probability, of order $N^{-1}$. If
$W_N \ll N^{5/6}$, the inverse participation ratio averaged over eigenvalues in a window
$[1 + a/W_N^{4/5}, 1 + b/W_N^{4/5}]$ is, with high probability, of order $W_N^{-6/5}$.

\rem4 We remark that Schenker \cite{Sch} proved a lower bound $\frac{1}{C W^{8}}$ on the inverse
participation ratio of the eigenvalues in the bulk of the spectrum (for a slightly
different class of band matrices). In the opposite direction, Erd\H{o}s and
Knowles \cite{EK1,EK2} recently proved an upper bound $W^{-\frac{1}{3} + o(1)}$
for a wide class of band matrices; their argument uses in particular the expansion in
Chebyshev polynomials developed in the current paper.

\rem4 Finally, there is a natural extension of band matrices to higher-dimen\-sional
lattices: the rows and columns of $H_N$ are indexed by elements of $(\ZZ/N\ZZ)^d$,
and $H_N(u, v) = 0$ unless $0 < \|u-v\| \leq W_N$. Similar random matrices have been
also studied in physical and mathematical literature, cf.\ Silvestrov \cite{Si},
Disertori, Pinson, and Spencer \cite{DPS}. We hope to consider the spectral edges of
such matrices in a forthcoming work.

\vspace{4mm}\noindent
{\bf Acknowledgment.} I am grateful to my supervisor, Vitali Milman, for his encouragement
and support.

Thomas Spencer has shared with me his interest in band matrices, and encouraged to apply
the method of \cite{FS} to study their spectral edges. His comments on a preliminary version
of this paper have been of great help. Yan Fyodorov has explained to me what is the
Thouless criterion, and why are its predictions coherent with the results of the current
paper. Bo'az Klartag has suggested to use the Cauchy--Binet formula to simplify the expressions
in Section~\ref{s:emb}. My father has referred me to the works \cite{L,LM,Ma,V}. Mark Rudelson
and Valentin Vengerovsky have helped fix the ambiguities in the definition of diagram, see
Remark~\ref{rem:clar}. The discussions with Alexei Khorunzhiy, Leonid Pastur, Mariya Shcherbina,
and Uzy Smilansky on related topics have been of great benefit to me.

I thank them all very much.

\end{document}